\documentclass[letterpaper, 10 pt, conference]{ieeeconf} 
\IEEEoverridecommandlockouts                              
\usepackage{amsmath}
\usepackage{mathtools}
\usepackage{amsfonts}
\usepackage{amssymb}
\usepackage{algorithm} 
\usepackage{algpseudocode}
\usepackage{color}
\usepackage{graphicx}
\usepackage{balance}

\newtheorem{definition}{Definition}

\newtheorem{remark}{\bf Remark}
\newtheorem{assumption}{Assumption}
\newtheorem{theorem}{Theorem}

\newcommand{\xddots}{%
	\raise 4pt \hbox {.}
	\mkern 10mu
	\raise 1pt \hbox {.}
	\mkern 10mu
	\raise -2pt \hbox {.}
}

\title{\LARGE \bf Data-driven identification of dissipative linear models \\ for nonlinear systems}
\author{S. Sivaranjani, Etika Agarwal and Vijay Gupta
	\thanks{S. Sivaranjani, Etika Agarwal and Vijay Gupta are with the Department of Electrical Engineering, University of Notre Dame, South Bend, IN. \{sseethar@nd.edu, etika.agarwal09@gmail.com, vgupta2@nd.edu\}. 
	}
}

\begin{document}
	\maketitle
	\thispagestyle{empty}
	\pagestyle{empty}
	
	\begin{abstract}
		We consider the problem of identifying a dissipative linear model of an unknown nonlinear system that is known to be dissipative, from time domain input-output data. We first learn an approximate linear model of the nonlinear system using standard system identification techniques and then perturb the system matrices of the linear model to enforce dissipativity, while closely approximating the dynamical behavior of the nonlinear system. Further, we provide an analytical relationship between the size of the perturbation and the radius in which the dissipativity of the linear model guarantees local dissipativity of the unknown nonlinear system. We demonstrate the application of this identification technique to the problem of learning a dissipative model of a microgrid with high penetration of variable renewable energy sources.
	\end{abstract}
	
	\begin{keywords}
		Dissipativity, identification, nonlinear systems, learning, passivity. 
	\end{keywords}
	
	\section{Introduction}
	The fields of system identification and control design initially developed in isolation \cite{gevers2005identification}. However, two systems that are `close' to each other in terms of the input-output response in the open loop may yield very different performance when put in feedback with the same controller. This realization led to the development of the area of identification for control, where the goal is to identify models such that controllers designed based on these models provide specific performance guarantees on the true system (see \cite{gevers2005identification} for a comprehensive survey of this area). Many such methods were developed over the last few decades, the most popular of which are iterative development of the system model and the controller \cite{aastrom1994analysis}-\nocite{de1997suboptimal}\nocite{lee1993new}\nocite{gevers1993towards}\nocite{hjalmarsson1996model}\cite{zang1995iterative}, and the development of data-based uncertainty sets for robust control \cite{bombois2000measure}-\nocite{makila1995worst}\nocite{de1995quantification}\cite{kosut1994least}. With the recent emergence of learning-based controller design, this field has seen a resurgence of interest as well. An important challenge that still remains open in this area is that of ensuring analytical guarantees on stability and performance of the closed loop system, with controllers that are designed based on models that are learned from data. 
	
	In this paper, we consider the following problem. Assume that we have access to some information about the true system satisfying a structural property that makes it easy to design a controller and obtain a desired performance or stability guarantee on the closed loop system. Can we identify a system model that satisfies this property? In particular, here, we consider the property to be that of dissipativity. Dissipativity is an important input-output property of dynamical systems \cite{antsaklis2013control} which encompasses many important special cases like $\mathcal{L}_2$ stability, passivity and conicity. Dissipativity, thus, finds application in various domains ranging from robotics \cite{hatanaka2015passivity}, electromechanical systems \cite{ortega2013passivity} and aerospace systems \cite{sivaranjani2015passivity}, to process control \cite{tippett2013distributed}\cite{bao2007process}, networked control and cyberphysical systems~\cite{agarwal2019distributed}-\nocite{sztipanovits2011toward}\cite{zhao2014feedback} and energy networks \cite{agarwal2017feedback}-\nocite{sivaranjani2018conic}\nocite{sivaranjani2018mixed}\cite{sivaranjani2018distributed}. Dissipative systems possess several desirable properties like stability and compositionality over certain interconnections \cite{antsaklis2013control}. Hence, if the original system is known to be dissipative, and we could exploit this fact to learn dissipative models, these models can then be used to design controllers that provide desired stability and performance guarantees on the original system. Note that existing methods in system identification may not yield a dissipative model even if the system is known to be dissipative. Furthermore, even if the model is dissipative, the dissipativity properties of the model do not, in general, yield any guarantee on the dissipativity properties of the true system. %In this paper, we provide an algorithm to identify a dissipative model of such a system, such that the dissipativity  of the model provides guarantees on the dissipativity of the true system as well. %Such a passive model can then be used, among other purposes, for designing a controller that provides stability and performance guarantees for the closed loop system when used on the true system as well.

	%In several critical applications such as the power grid, where high-order nonlinear models are commonly encountered, it may not be possible to provide provable stability and performance guarantees using controllers designed based on purely data-driven models. However, if the original system is known to be passive, this fact may be exploited to learn passive linear models, which can then be used to design controllers that provide desired stability and performance guarantees on the original nonlinear system. Furthermore,  approximate linear dynamical models are already known from the physics in most application domains, or easily obtained by standard system identification techniques \cite{katayama2006subspace}\cite{verhaegen2007filtering}. Therefore, the insight into the structure and physics of these models can be exploited to estimate a control-oriented passive dynamical model of the system.
	
	% \textcolor{blue}{As stated earlier, existing methods in system identification may not yield a dissipative model even if the system is known to be dissipative.} 
	
	In this paper, we solve this problem of identifying a dissipative linear model of an unknown dissipative nonlinear dynamical system from given time-domain input-output data. Inspired by passive macromodeling approaches from RF circuit theory \cite{grivet2015passive}, we propose a two-stage approach. First, we learn an approximate linear model of the system, referred to as a baseline model, either using standard system identification techniques or using physics-based knowledge of the system. Next, we perturb the system matrices of this baseline linear model to enforce quadratic (QSR) dissipativity. We show that this perturbation can be chosen to ensure that the input-output behavior of the dissipative linear approximation closely approximates that of the original nonlinear system, provided that the baseline linear model closely approximates the nonlinear system dynamics in the input-output sense. Further, we provide an analytical condition relating the size of the perturbation to the radius in which local quadratic dissipativity properties of the nonlinear system can be guaranteed by the dissipative linear model. This relationship formalizes the intuition that larger perturbations lead to poorer approximations; in other words, the radius of local dissipativity of the nonlinear system decreases as the size of the perturbation is increased. Finally, we demonstrate the application of this approach to the problem of learning a dissipative model towards control of a microgrid with high penetration of renewable energy sources.
	
	We remark that if the main objective is simply to learn the passivity index of the system, which can be considered a specific dissipativity property, recently developed allied approaches can be utilized to directly learn the index from input-output data \cite{montenbruck2016some}-\nocite{romer2017sampling}\cite{zakeri2019data}. In contrast to these works, our approach can be used to learn a broader class of dissipative models, encompassing properties like passivity, sector boundedness and $\mathcal{L}_2$ stability. In addition, our approach yields a model with guarantees on the dissipativity and the input-output response of the original system. Further, there is a stream of work that relates passivity of a system to its approximation~\cite{xia2015determining}\cite{xia2016passivity}; however, in that stream, a (dissipative) model of the system is assumed to be present, and is also assumed to be the first order Taylor approximation of the nonlinear system, which may not be the case for models identified from data. 
	
	This paper is organized a follows. In Section \ref{sec:sys_dyn}, we introduce the system model and formally state the problem addressed in this paper. In Section \ref{sec:sol}, we present the two-stage approach to learning linear dissipative models for unknown nonlinear systems. %, and provide an analytical quantification of the tradeoff between the perturbation size and local dissipativity properties of the model. 
	In Section \ref{sec:case}, we demonstrate this approach numerically. 
	
	\vspace{0.5em}
	\textit{Notation:} We denote the sets of real numbers, positive real numbers including zero, and $n$-dimensional real vectors by $\mathbb{R}$, $\mathbb{R}_{+}$ and $\mathbb{R}^{n}$ respectively. %Define $\mathcal{N}_N=\{1,\ldots,N\}$, where $N$ is a natural number excluding zero. 
	Given a matrix $A \in \mathbb{R}^{m\times n}$, $A'\in \mathbb{R}^{n \times m}$ represents its transpose. %, and if all the eigenvalues of $A$ are real, then $\lambda_{max}(A)$ represents its maximum eigenvalue. 
	A symmetric positive definite matrix $P \in \mathbb{R}^{n \times n}$ is represented as $P>0$ (and as $P\geq 0$, if it is positive semi-definite). The standard identity matrix is denoted by $I$, and a matrix with all elements equal to 1 is denoted by $\mathbf{1}$, with dimensions clear from the context. %Given sets $A$ and $B$, $A\backslash B$ represents the set of all elements of $A$ that are not in $B$. 
	Given a function $f$, $\text{dom} f$ represents its domain.
	
	\section{Problem Formulation}\label{sec:sys_dyn}
	We consider an unknown nonlinear dynamical system
	\begin{equation}
	\label{eq:nonlinear}
	\begin{aligned}
	S_{nl}: \quad  &\dot x(t)=f(x(t),u(t))\\
	&y(t)=g(x(t),u(t)),
	\end{aligned}
	\end{equation}
	where $f$ and $g$ are differentiable functions, and $x(t) \in \mathbb{R}^n$, $u(t)\in \mathbb{R}^m$ and $y(t)\in \mathbb{R}^p$ represent the state, input and output of the system at time $t\in \mathbb{R}_{+}$ respectively.
	\begin{assumption}
		The functions $f$ and $g$ are Lipschitz continuous, that is, 
		\begin{equation}\label{eq:lipschitz}
		\begin{aligned}
		||f(a_1)-f(a_2)||&\leq L_f ||a_1-a_2||, \quad \forall a_1,a_2 \in \text{dom} f\\
		||g(a_1)-f(a_2)||&\leq L_g ||a_1-a_2||, \quad \forall a_1,a_2 \in \text{dom} g,
		\end{aligned}
		\end{equation}
		where $L_f$ and $L_g$ are the Lipschitz constants of $f$ and $g$ respectively.
		\begin{assumption}
			There exists an equilibrium point $(x^{*},u^{*})=(0,0)$ for system \eqref{eq:nonlinear} such that $f(x^{*},u^{*})=0$.
		\end{assumption}
	\end{assumption}
	
	Note that the assumption of an equilibrium point at the origin is sufficiently general since the system dynamics around a non-zero equilibrium can be obtained by a suitable coordinate transformation. The following definition of dissipativity is standard for such systems; however, we also define the notion of strict dissipativity as follows. 
	\begin{definition}[Dissipativity and Strict Dissipativity]\label{def:passivity}
		%Consider all states $x \in \mathcal{X}$ and control inputs $u \in \mathcal{U}$, where $\mathcal{X}\times \mathcal{U}$ is a neighborhood of the equilibrium (origin) $(x^{*},u^{*})=0$. 
		Let $\mathcal{X}\times \mathcal{U}$ be a neighborhood of the equilibrium (origin) $(x^{*},u^{*})=0$. The nonlinear system $S_{nl}$ is said to be (locally) \textit{dissipative} with dissipativity matrices $Q=Q'$, $S$ and $R=R'$, if
		% such that 
		\begin{equation}\label{eq:diss}
		y'(t)Qy(t)+u'(t)Ru(t)+2y'(t)Su(t) \geq 0,
		\end{equation}
		$\forall t\in \mathbb{R}_{+}$ and $\forall x \in \mathcal{X}$ and control inputs $u \in \mathcal{U}$. Further, the system $S_{nl}$ is said to be (locally) \textit{strictly dissipative} (SD) with dissipativity matrices $Q=Q'$, $S$ and $R=R'$, if there exist constants $\rho>0$ and $\nu>0$, referred to as dissipativity indices, such that
		\begin{multline}\label{eq:passivity}
		%   \begin{aligned}
		y'(t)Qy(t)+u'(t)Ru(t)+2y'(t)Su(t) \\ \geq \rho x'(t)x(t) +\nu u'(t)u(t)
		%\end{aligned} 
		\end{multline}
		$\forall t\in \mathbb{R}_{+}$ and $\forall x \in \mathcal{X}$ and control inputs $u \in \mathcal{U}$. %Further, the system \eqref{eq:nonlinear} is said to be \textit{locally dissipative} (or \textit{locally SD}) if \eqref{eq:diss} (or \eqref{eq:passivity} holds $\forall x \in \mathcal{X}$ and $u \in \mathcal{U}$, where $\mathcal{X}\times \mathcal{U}$ is a neighborhood of the equilibrium (origin) $(x^{*},u^{*})=0$.
	\end{definition}
	
	We ignore the qualifier `locally' in front of dissipativity properties for pedagogical ease. For the remainder of the paper, we also drop the dependence of all vectors on time for simplicity of notation.
	\begin{remark}\label{rem:qsr}
		Definition \ref{def:passivity} represents the property of quadratic dissipativity, commonly referred to as $QSR$-dissipativity in literature \cite{agarwal2018passivity}. We choose this specific class of dissipativity, since it can be used to capture several useful system properties through appropriate choice of the dissipativity matrices $Q$, $S$ and $R$ in \eqref{eq:diss} such as:
		\begin{enumerate}
			\item[(i)] \textit{passivity}, with $Q = 0$, $S = \frac{1}{2}I$ and $R =0$,
			\item[(ii)] \textit{strict passivity}, with $Q=-a I$, $S= \frac{1}{2}I$ and $R =-b I$, where $a,b \in \mathbb{R}^+\backslash\{0\}$,
			\item[(iii)] \textit{$\mathcal{L}_2$ stability}, with $Q=-\frac{1}{\gamma} I$, $S= 0$ and $R =\gamma I$ where $\gamma\in \mathbb{R}^+$ is an $\mathcal{L}_2$ gain of the system, 
			\item[(iv)] \textit{conicity}, with $Q=-I$, $S = c I$ and $R = (r^2-c^2)I$, where $c\in\mathbb{R}$ and $r \in \mathbb{R}^+\backslash\{0\}$, and,
			\item[(v)] \textit{sector-boundedness}, with $Q=-I$, $S = (a+b)I$ and $R = -ab I$, where $a,b\in \mathbb{R}$.
		\end{enumerate}
	\end{remark}
	Note that any SD system is also dissipative and satisfies \eqref{eq:diss}. 
	We now formally state the identification problem addressed in this paper. 
	
	\textbf{Problem $\mathcal{P}$:} Given a set of time domain input-output measurements $(\hat y, \hat u)$ from a dissipative nonlinear system $S_{nl}$ satisfying \eqref{eq:diss}, the aim of this paper is to obtain a linear model
	\begin{equation}
	\label{eq:linear}
	\begin{aligned}
	S_l: \quad     &\dot {\tilde x}=A\tilde x+Bu\\
	&\tilde y=C\tilde x+Du,
	\end{aligned}
	\end{equation}
	such that 
	\begin{enumerate}
		\item[(i)] $||\tilde y-\hat y||^2_2 <\tilde \delta_y$ when $u\in\mathcal{U}$, and
		\item[(ii)] $S_l$ is strictly dissipative. %, that is Definition \ref{def:passivity} holds for $S_l$ with $y(t) \mapsto \tilde y(t)$.
	\end{enumerate}
	\begin{figure}[b]
		\centering
		% \vspace{0.6em}
		\includegraphics[scale=0.5]{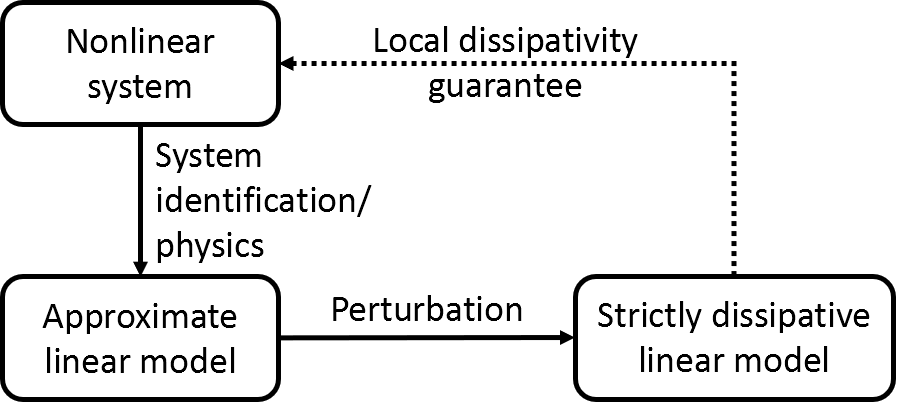}
		\caption{Schematic of two-stage approach for identification of dissipative linear models.}
		\label{fig:schematic}
	\end{figure}
	
	We will address this problem in two stages, as shown in Fig. \ref{fig:schematic}. We will begin by assuming that an approximate linear model of $S_{nl}$ can be estimated either through standard regression-based or subspace system identification, and/or from the physics of the system. If this linear approximation is not SD, we will then introduce a bounded perturbation into the system matrices such that the resulting perturbed model is SD, while closely approximating the behavior of the nonlinear system $S_{nl}$. We require the linear model $S_l$ to be strictly dissipative rather than just dissipative, since this allows us to provide guarantees on the local dissipativity of the original system $S_{nl}$.
	
	We conclude this section by stating a matrix inequality that can be used to verify if the linear model $S_l$ is SD.
	\begin{theorem}\cite{xia2015determining}
		The linear system \eqref{eq:linear} is strictly dissipative if there exists a symmetric matrix $P=P'$, and constants $\nu>0$ and $\rho>0$ satisfying 
		\begin{equation}\label{eq:ssip_lmi}
		\begin{bmatrix}
		A'P+PA-C'QC+\rho I & PB-\hat S\\
		B'P-\hat S' & -\hat R+\nu I
		\end{bmatrix}  <0, 
		\end{equation}
		where $\hat S=C'S+C'QD$ and $\hat R=R+D'S+S'D+D'QD$. Further, if this condition is satisfied, the linear system is {\em globally} dissipative.
	\end{theorem}

	\section{Identification of Dissipative Models}\label{sec:sol}
	In this section, we describe a two-stage approach to identify a dissipative linear model $S_l$ that closely approximates the nonlinear system $S_{nl}$. This approach is inspired by similar perturbation approaches used to obtain passive macromodels in RF electronics literature \cite{grivet2015passive}.
	\subsubsection*{Baseline linear model} Given a set of time domain input-output measurements $(\hat y, \hat u)$, $\hat u \subset \mathcal{U}$ from system $S_{nl}$ in the vicinity of the equilibrium, we begin by assuming that a standard  identification technique \cite{katayama2006subspace}\cite{verhaegen2007filtering} can be used to identify an approximate linear model,
	\begin{equation}
	\label{eq:app_linear}
	\begin{aligned}
	S_b: \quad     &\dot {\bar x}=\bar A\bar x+\bar Bu\\
	&\bar y=\bar C\bar x+\bar Du,
	\end{aligned}
	\end{equation}
	referred to as the \textit{baseline linear model}, such that $||\bar y-\hat y||^2_2 <\bar \delta_y$, $\forall u\in \mathcal{U}$. 
	
	We also estimate the Lipschitz constant $L_g$ of $S_{nl}$ as 
	\begin{equation}\label{eq:lipschitz}
	L_g=\max_{u_1,u_2 \in \hat u, u_1\neq u_2} \frac{||y_1-y_2||}{||u_1-u_2||},
	\end{equation}
	where $y_1$ and $y_2$ are the outputs of $S_{nl}$ corresponding to the inputs $u_1$ and $u_2$ respectively. Alternatively, \eqref{eq:lipschitz} can be applied to the approximate linear system $S_b$ to easily obtain an estimate of the Lipschitz constant $L_g$. Note that it has been observed that \eqref{eq:lipschitz} provides a good estimate of the Lipschitz constant of $S_{nl}$ if the data set $(\hat y, \hat u)$ is sufficiently rich \cite{montenbruck2016some}.
	
	\subsubsection*{Perturbed linear model} If the linear model $S_b$ is not SD, that is, it does not satisfy \eqref{eq:passivity} with $y$ replaced by $\bar y$, then, we would like to introduce a bounded perturbation $\Delta C$ into the output matrix of $S_b$ to obtain the perturbed linear model
	\begin{equation}
	\label{eq:app_linear}
	\begin{aligned}
	S_l: \quad     &\dot {\tilde x}= A\tilde x+ Bu\\
	&\tilde y= C\tilde x+ Du,
	\end{aligned}
	\end{equation}
	where $A=\bar A$, $B=\bar B$, $C=\bar C +\Delta C$ and $D=\bar D$. 
	\begin{remark} We have chosen to perturb the output matrix $\bar C$ to obtain the perturbed linear model $S_l$. However, we make the following comments. 
		\begin{enumerate}
			\item[(i)] The input matrix $\bar B$ or the feedforward matrix $\bar D$ can be perturbed instead of the output matrix $\bar C$, depending on system specific requirements. 
			\item[(ii)] Any perturbation on the system matrix $\bar A$ is not preferable, since we would like the perturbed linear model to preserve any information about the dominant modes of the nonlinear system that is embedded in the baseline linear model, thereby allowing the perturbed model to closely approximate the original nonlinear system. 
			\item[(iii)] If the baseline model $S_b$ has $\bar D=0$ and it is required to ensure $D>0$ in the linear model $S_l$ to meet strict dissipativity or other desired system properties, then the feedforward matrix $\bar D$ can be perturbed to enforce the positive definiteness of $D$. 
		\end{enumerate}
	\end{remark}
	
	We would like to minimize the size of the perturbation $||\Delta C||_2^2$, in order to ensure that the linear model $S_l$ closely approximates the original nonlinear system $S_{nl}$. Further, we would like to relate the strict dissipativity of $S_l$ to local  dissipativity of the nonlinear system $S_{nl}$. We have the following result on the choice of the perturbation $\Delta C$, and its relationship to the strict dissipativity of $S_l$ and $S_{nl}$. 
	\begin{theorem}\label{thm:pert}
		Given the linear model \eqref{eq:app_linear}, if problem
		\begin{subequations}
			\begin{align} 
			\mathcal{P}_1: &\quad \min \limits_{\nu>0,\rho>0,P>0,\Delta C} \alpha=||\Delta C||^2_2\\
			\text{s.t.} &\quad \begin{bmatrix}
			A'P+PA-C'QC+\rho I & PB-\hat S \\
			B'P-\hat S' & -\hat R+\nu I
			\end{bmatrix}  <0,  \label{eq:p2}  \quad \nonumber\\
			&\quad \hat S=C'S+C'QD, \quad \nonumber \\ &\quad \hat R=R+D'S+S'D+D'QD \quad \nonumber \\
			&\quad P=P'>0\end{align}\vspace{-2.5em}\begin{align}  %\label{eq:p3}
			\quad \quad \rho &\geq ||Q||^2||\bar C+\Delta C||^2 \quad \nonumber \\&+2\left(||Q||^2+I\right) (Lg+||\bar C +\Delta C)||)^2  \label{eq:l1} \end{align}\vspace{-2.3em}\begin{align}
			\quad \quad\quad  \nu &\geq ||S+QD||^2+ 2\left(||Q||^2+I\right) (Lg+||D||)^2 \label{eq:l2}
			\end{align}
		\end{subequations}
		
		is feasible, then
		\begin{enumerate}
			\item[(i)] $S_l$ is SD,
			\item[(ii)] $S_{nl}$ is locally dissipative in a neighborhood $\mathcal{X}\times \mathcal{U}$ around the origin, and
			\item[(iii)] $S_l$ closely approximates $S_{nl}$, that is $$||\tilde y-\hat y||^2_2 <\tilde \delta_y=(1+\beta)\bar \delta_y, \beta\geq 0, \forall u\in \mathcal{U}.$$
		\end{enumerate}
	\end{theorem}
	\begin{proof}
		We separately prove each part of Theorem \ref{thm:pert}. 
		\begin{enumerate}
			\item[(i)] If $\mathcal{P}_1$ is feasible, then \eqref{eq:p2} is satisfied for the dynamics \eqref{eq:linear}. Therefore, from \eqref{eq:passivity}, $S_l$ is SD.
			\item[(ii)] Define the error in the input-output response between the linear model $S_l$ and the nonlinear system $S_{nl}$ as
			\begin{equation}
			\label{eq:epsilon}
			\begin{aligned}
			% \epsilon_f &= A\tilde x+ Bu-f(x,u)\\
			\epsilon_g &= C\tilde x+ Du-g(x,u).
			\end{aligned}
			\end{equation}
			Then, we have
			% \begin{equation}\label{eq:a1}
			$y=\tilde y-\epsilon_g.$ 
			% \end{equation}
			Now, if $\mathcal{P}_1$ is feasible, $S_l$ is SD and satisfies \eqref{eq:passivity}. If \eqref{eq:passivity} holds for any $\tilde x$, then it must also hold for $\tilde x=x$.  Therefore, we have
			\begin{equation}\label{eq:a2}
			\tilde y'Q\tilde y+ u'Ru +2 \tilde y'Su \geq \rho ||x||^2+\nu ||u||^2. 
			\end{equation}
			Also, from \eqref{eq:epsilon} and \eqref{eq:lipschitz}, %we have
			% $$\epsilon_g =-g(x,u)+g(0,0)+Cx+Du.$$
			%Since $g(0,0)=0$, 
			we can write
			\begin{equation}\label{eq:a5}
			\begin{aligned}
			||\epsilon_g||&\leq L_g ||x||+L_g ||u||+||C|| ||x||+||D|| ||u||\\
			&=\left(L_g+||C||\right)||x||+\left(L_g+||D||\right)||u||.
			\end{aligned}
			\end{equation}
			Using Jensen's inequality in \eqref{eq:a5} gives 
			\begin{equation}\label{eq:a6}
			||\epsilon_g||^2\leq 2(L_g+||C||)^2||x||^2+2(L_g+||D||)^2||u||^2.
			\end{equation}
			
			Now consider 
			\begin{equation}
			\begin{aligned}\label{eq:I}
			I&=y'Qy+2y'Su+u'Ru\\
			&=(\tilde y -\epsilon_g)Q(\tilde y -\epsilon_g)+2(\tilde y -\epsilon_g)'Su+u'Ru\\
			&=\phi- 2\epsilon_g'Q\tilde y -2 \epsilon_g'S u
			\end{aligned},
			\end{equation}
			where $\phi=\tilde y'Q\tilde y+ u'Ru +2 \tilde y'Su - \epsilon_g'Q\epsilon_g.$ Then, from \eqref{eq:a2} and \eqref{eq:a6}, we have,
			\begin{equation}\label{eq:a7}
			\begin{aligned}
			\phi&\geq\left(\rho-2||Q||^2(L_g+||C||)^2\right)||x||^2\\&+\left(\nu-2||Q||^2(L_g+||D||)^2\right)||u||^2.
			\end{aligned}
			\end{equation}
			
			We also have
			\begin{equation}\label{eq:a8}
			2\epsilon_g'Q\tilde y +2 \epsilon_g'S u=2\epsilon_g'QCx+2\epsilon_g'(S+QD)u,
			\end{equation}
			where
			\begin{equation}
			\begin{aligned}\label{eq:a9}
			2\epsilon_g'QCx &\leq ||\epsilon_g||^2+||Q||^2||C||^2||x||^2,\\
			% &\leq (2\lambda_{max}(Q'Q)(L_g+||C||)^2+\lambda_{max}(C'C))||x||^2 \\
			% &+2\lambda_{max}(Q'Q)(L_g+||C||)^2||u||^2.
			% \end{aligned}
			% \end{equation}
			% and
			% \begin{equation}\label{eq:a10}
			% \begin{aligned}
			2\epsilon_g'(S+QD)u &\leq ||\epsilon_g||^2+||(S+QD)||^2||u||^2.
			% &\leq (2\lambda_{max}(Q'Q)(L_g+||C||)^2+\lambda_{max}(C'C))||x||^2 \\
			% &+2\lambda_{max}(Q'Q)(L_g+||C||)^2||u||^2.
			\end{aligned}
			\end{equation}
			
			If $\mathcal{P}_1$ is feasible, then \eqref{eq:l1} and \eqref{eq:l2} hold. Then, using \eqref{eq:l1}, \eqref{eq:l2} and \eqref{eq:a7}-\eqref{eq:a9} in \eqref{eq:I}, we have
			\begin{equation}\label{eq:I2}
			I \geq \hat \rho ||x||^2 +\hat \rho||u||^2 \geq 0,
			\end{equation}
			with 
			$  \hat \nu >0, \quad \hat \rho >0,$
			where
			\begin{equation*}
			\begin{aligned}
			\hat \rho =& \rho -||Q||^2||\bar C+\Delta C||^2\\
			&-2\left(||Q||^2+I\right) (Lg+||\bar C +\Delta C)||)^2,\\
			\hat \nu =& \nu -2\left(||Q||^2+I\right) (Lg+||D||)^2-||S+QD||^2.
			\end{aligned}
			\end{equation*}
			Using \eqref{eq:I2} in Definition \ref{def:passivity}, $S_{nl}$ is locally dissipative in a neighborhood $\mathcal{X}\times \mathcal{U}$ of the origin if $\mathcal{P}_1$ is feasible, where $\mathcal{X}\times \mathcal{U}$ is an $\epsilon$-ball around the origin, with
			\begin{equation}\label{eq:ball}
			\epsilon = \min \left(\frac{\epsilon_g}{\sqrt{2}(L_g+||\bar C+\Delta C||)},\frac{\epsilon_g}{\sqrt{2}(L_g+||\bar D||)}\right).
			\end{equation}
			\item[(iii)] If $\mathcal{P}_1$ is feasible, then, for the baseline model $S_b$ with $\bar u = u \in \mathcal{U}$, we have
			% \begin{equation}\label{eq:a9}
			$||\bar y-\hat y||^2_2 <\bar \delta_y.$ 
			% \end{equation}
			Then, we can write
			\begin{equation*}
			\begin{aligned}
			||\tilde y-\hat y||^2_2 &= ||\tilde y-\hat y+\bar y -\bar y||^2_2\\
			&\leq ||\tilde y-\bar y||^2_2 + ||\bar y-\hat y||^2_2 \\
			&\leq ||\tilde y-\bar y||^2_2 +\bar \delta_y \\
			%    &\leq \alpha ||x||^2 +\bar \delta_y\\
			&\leq \alpha \epsilon^2 +\bar \delta_y=(1+\beta)\delta_y, \quad \beta={\alpha \epsilon^2}/{\bar \delta_y}.
			\end{aligned}
			\end{equation*}
		\end{enumerate}
		\vspace{-0.5em}
	\end{proof}
	
	Theorem \ref{thm:pert} provides conditions that can be used to choose the perturbation such that the linear model obtained closely approximates the original nonlinear system. Further, if $\mathcal{P}_1$ is feasible, then the nonlinear system $S_{nl}$ is strictly dissipative in a neighborhood around the origin. The process of identifying a linear model $S_l$ that solves problem $\mathcal{P}$ is provided in Algorithm 1.
	% \begin{subfigures}
	% \begin{figure}
	% \vspace{-0.45em}
	% \end{figure}
	
	\begin{algorithm}[H]
		\caption{Identification of dissipative model}
		\label{alg:fitting}
		\hspace*{\algorithmicindent} \textbf{Input} Measurement vectors $\{\hat y\}$ and $\{\hat u\}$.  \\
		\hspace*{\algorithmicindent} \textbf{Output} $A$, $B$, $C$, $D$ and $\alpha$.
		\begin{algorithmic}[1]
			\State \textbf{Estimate baseline model:} Use standard subspace or regression-based identification techniques \cite{katayama2006subspace}\cite{verhaegen2007filtering} to estimate $\bar A$, $\bar B$, $\bar C$, $\bar D$ of $S_b$ such that $||\bar y-\hat y||_2^2$ is minimized. 
			\State \label{compute} Check if ${\mathcal{P}_2}$ is feasible, where 
			\begin{eqnarray*} \label{eq:p1}
				\mathcal{P}_2: &\quad \mbox{Find:} \quad {\nu>0,\rho>0,P>0} \\
				\text{s.t.} &\quad \begin{bmatrix}\label{eq:pert_ssip}
					\bar A'P+P\bar A-\bar C'Q\bar C+\rho I &P\bar B-\hat S \\
					\bar B'P-\hat S'+ & -\hat R+\nu I
				\end{bmatrix}\leq 0,\\
				&\hat S=\bar C'S+\bar C'Q\bar D,\\& \hat R=R+\bar D'S+S'\bar D+\bar D'Q\bar D.
			\end{eqnarray*}
			\If{$\mathcal{P}_2$ is feasible,} 
			\State Set $A=\bar A$, $B=\bar B$, $C=\bar C$, $D=\bar D$.
			\Else
			\State \textbf{Perturbation model:} Set $C_i \gets \bar C + \Delta C$, where $\Delta C=\gamma \mathbf{1}$. 
			\State Find $\nu>0,\rho>0, P>0$ and $\gamma>0$ solving $\mathcal{P}_1$ with constraints \eqref{eq:p2} and \eqref{eq:l2}. \label{alg:s1}
			
			\If{$\rho$ from Step \ref{alg:s1} satisfies constraint \eqref{eq:l1}}
			\State Set $A=\bar A$, $B=\bar B$, $C=\bar C+\Delta C$, $D=\bar D$.
			\State Compute $\alpha$.
			\Else
			\State Increase $\rho\mapsto \rho+d$, where $d>0$.
			\State Go to Step \ref{alg:s1}. 
			\EndIf
			\EndIf
		\end{algorithmic}
		% 	\vspace{0.5em}
	\end{algorithm}
	% \end{subfigures}
	%%%%%%%%%%%%%%%%%%%%%%%%%%%%%%%%%%%%%%%%%%%%%%%%%%%%%%%%
	\begin{remark}
		We make the following observations about the results in Theorem \ref{thm:pert}.
		\begin{enumerate}
			\item[(i)] As the size of the perturbation  $||\Delta C||^2_2$ increases, the constraint \eqref{eq:l2} becomes harder to satisfy, that is, the model will require higher dissipativity indices, consequently resulting in a poorer fit. Therefore, we observe that  the error bound $(1+\beta)\bar \delta_y$ of the linear model $S_l$ grows with the size of the perturbation $\alpha$.
			\item[(ii)] Equation \eqref{eq:ball} provides a condition relating the size the perturbation to the radius in which local strict dissipativity of the nonlinear system $S_{nl}$ can be guaranteed by strict dissipativity of the linear model $S_l$. The $\epsilon$-neighborhood in which the local dissipativity of the nonlinear system is guaranteed shrinks with the size of the perturbation. Therefore, while large perturbation may be used to obtain a dissipative linear model of a nonlinear system, the radius of validity of this model and the radius of dissipativity of the nonlinear system would be extremely small.  
			\item[(iii)] The objective function of $\mathcal{P}_1$ is non-convex. This can be addressed by choosing a fixed perturbation $\Delta C=\gamma \mathbf{1}$, thereby transforming the objective function to 
			$$\min \limits_{\nu>0,\rho>0,P>0,\gamma} \gamma^2.$$
			\item[(iv)] The constraint \eqref{eq:l1} in $\mathcal{P}_1$ is non-convex. However, in practice, it is easy to solve $\mathcal{P}_1$ in two steps. First, we find some $\nu>0$ and $\rho>0$ such that $\mathcal{P}_1$ is feasible with constraints \eqref{eq:p2} and \eqref{eq:l2}. Then, we check if \eqref{eq:l1} is  feasible. If not, we increase the value of $\rho$ and re-solve $\mathcal{P}_1$. 
			\item[(v)] With a small perturbation, the dissipative linear model $S_l$ closely approximates the behavior of the original nonlinear system, provided that the error of the identification procedure used to obtain the baseline model is sufficiently small. 
			\item[(vi)] For specific cases of dissipativity, such as passivity, the constants $\rho$ and $\nu$ have special meaning since they can be interpreted as measures of levels of passivity through the concept of passivity indices.
		\end{enumerate}
	\end{remark}

	\section{Case Studies}\label{sec:case}
	In this section, we provide two numerical examples to illustrate the identification approach proposed in Section \ref{sec:sol}. %We then provide a more complex application of this approach to learning a dissipative model of a microgrid with large-scale renewable energy penetration.
	\subsubsection*{Example 1}
	We consider the nonlinear system 
	\begin{equation}\label{eq:learn_nonlinear}
	\begin{aligned}
	\dot x_1&= -x_1^2 +x_2, \quad 
	\dot x_2=-x_1-x_2+(0.5x_1+1)u\\
	y&=x_1+x_2+(0.5x_1+1)u.
	\end{aligned}
	\end{equation}
	It can be verified that \eqref{eq:learn_nonlinear} is dissipative, and more specifically, strictly passive in the sense of Remark \ref{rem:qsr}-(ii). Therefore, we would like to learn a linear model to reflect this property. 
	Following the procedure in Algorithm \ref{alg:fitting}, we first learn a baseline linear model of this system with system matrices
	\begin{equation}\label{eq:learn_baseline}
	\begin{aligned}
	\bar A&=\begin{bmatrix}
	0 & 1\\ -46.24&-22.31 
	\end{bmatrix}, \bar B=[0 \quad 1]',\\ \bar C&=[95.61 \quad -4.78], \quad \; \bar D=0.1
	\end{aligned}
	\end{equation}
	using the MATLAB System Identification Toolbox. The response of the baseline model and the training data used to obtain the model are shown in Fig. \ref{fig:training_model}. 
	We then verify that $\mathcal{P}_2$ in Step 2 of Algorithm \ref{alg:fitting} is not feasible with the baseline model \eqref{eq:learn_baseline}. Therefore, we follow the procedure outlined in the algorithm to obtain the perturbed linear model with $\Delta C =\gamma \mathbf{1}=9.53 \times [1 \quad 1]$. This perturbed linear model is strictly passive, and satisfies \eqref{eq:ssip_lmi} with the appropriate dissipativity matrices. We also observe that the linear model closely approximates the nonlinear system by validating its input-output response against a test data set (Fig. \ref{fig:passive_model}).
	\begin{figure}[!t]
		\centering
		\vspace{0.6em}
		{\includegraphics[scale=0.26,trim=0cm 0.5cm 0cm 0.5cm]{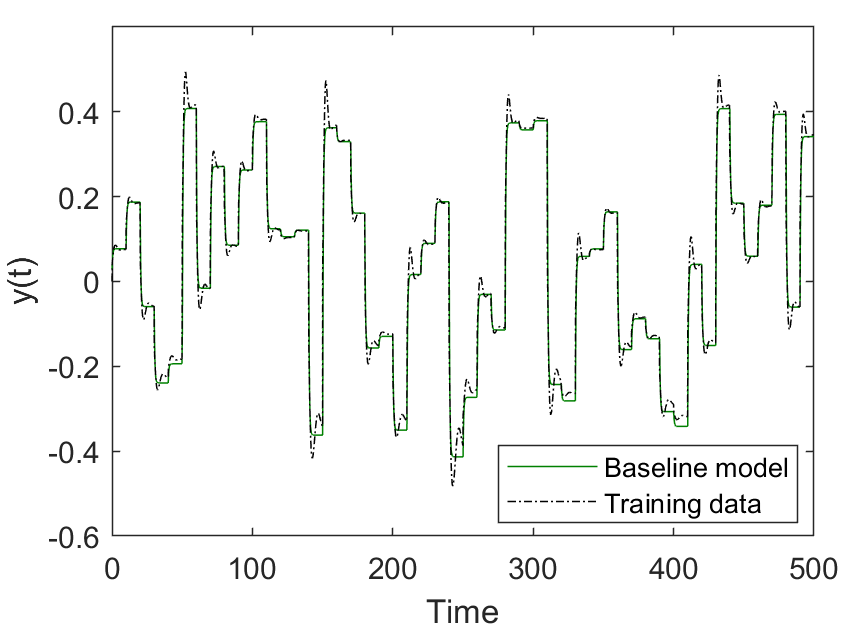}} \caption{Baseline model and training data for system.}\label{fig:training_model}
		\vspace{0.1em}
		{\includegraphics[scale=0.26,trim=0cm 0.5cm 0cm 0cm]{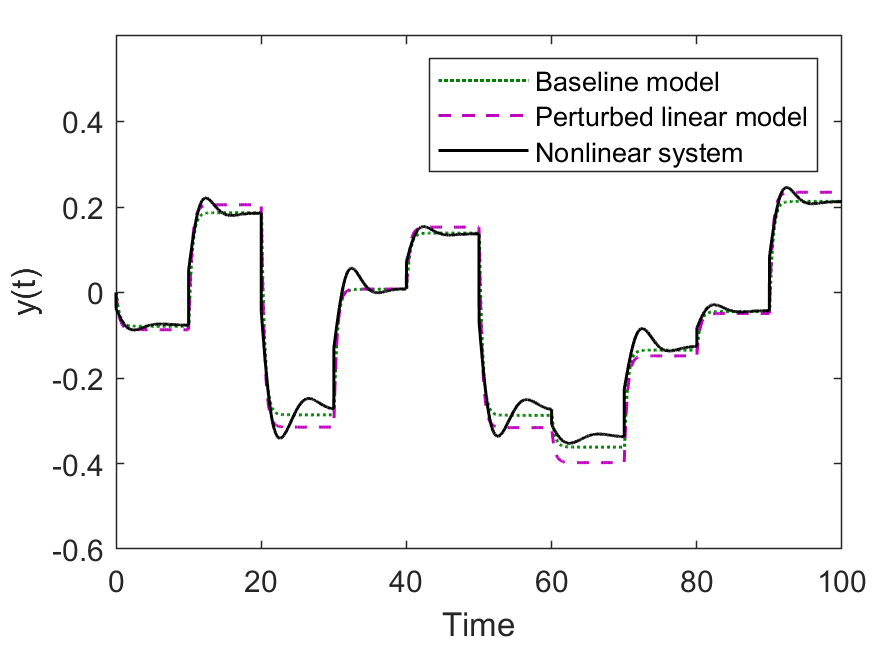}}  \caption{Input-output performance of the strictly passive linear model, the baseline model and the nonlinear system. \vspace{-2.3em}}\label{fig:passive_model} 
		%   \caption{Passive model identification for nonlinear system \eqref{eq:learn_nonlinear}.}\label{fig:example}
	\end{figure}
	\subsubsection*{Example II - Microgrid}
	We now consider the application of the proposed approach to obtain a dissipative model of the 14-bus microgrid system shown in Fig. \ref{fig:test_system}, in the vicinity of a specific power flow operating point (equilibrium). 
	The system shown in Fig. \ref{fig:test_system} is obtained as a modification the standard IEEE 14-bus test system \cite{ieee14} by replacing the largest generators in the system at buses 1, 2 and 3 with equivalent DFIG wind, photovoltaic and solid oxide fuel cell plants of 600 kVA, 60 kVA and 60 kVA respectively. The synchronous generators at buses 6 and 8 are rated 25 kVA each (see \cite{sivaranjani2013networked} detailed state space models of the system).
	\begin{figure}[b]
		\centering
		\vspace{-1.7em}
		\includegraphics[scale=0.3,trim=0cm 1cm 0cm 0cm]{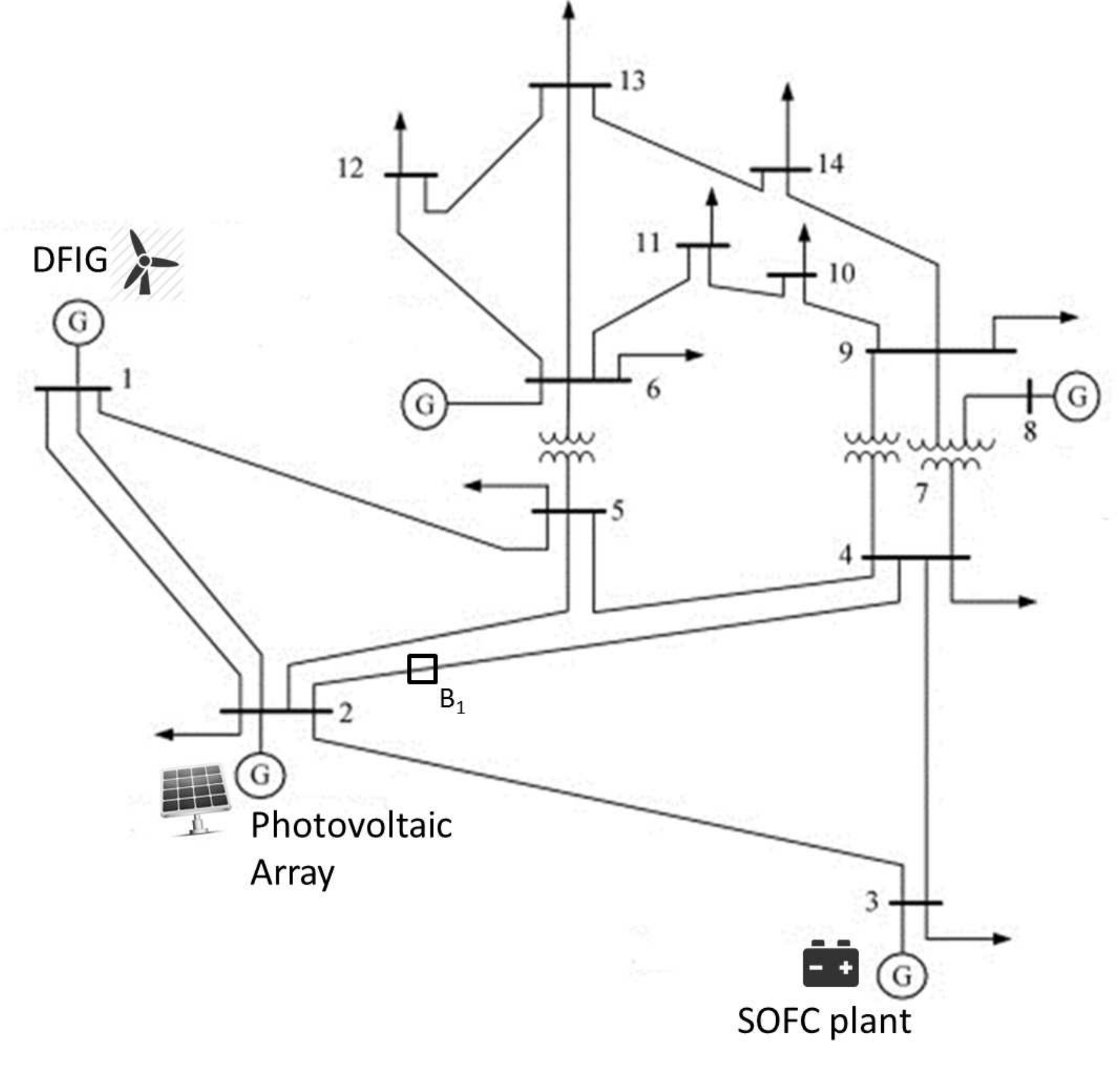}
		\caption{Example: 14-bus microgrid for dissipative model identification.}
		\label{fig:test_system}
	\end{figure}
	Therefore, 93.5\% of the generation in this system is attributed to renewable generators, making this system challenging to control. However, the system is known to be conic in the sense of Definition \ref{rem:qsr}-(iv), and this property can be exploited to design controllers that enhance the performance and stability of this system, even with the variability introduced by the renewable energy generators \cite{sivaranjani2018conic}\cite{agarwal2017feedback}. Therefore, we obtain a linear conic model of this system using the procedure described in Algorithm \ref{alg:fitting}. We note that a baseline model for this system can be readily obtained since the structure of the nonlinear differential equations, as well as estimates of the system parameters are well known from the system physics \cite{sivaranjani2013networked}. Figure~\ref{fig:test_compare} shows the comparison between the measured voltage outputs and those generated by the conic model at bus 1 (wind generator) for a load change (disturbance) where all loads in the network are decreased by 2\%. These results indicate that models with suitable dissipativity properties can be constructed to closely approximate the dynamics of complex nonlinear networked systems around specific operating points.
	\begin{figure}
		\centering
		\vspace{0.6em}
		\includegraphics[scale=0.14,trim=0cm 1cm 0.2cm 2cm]{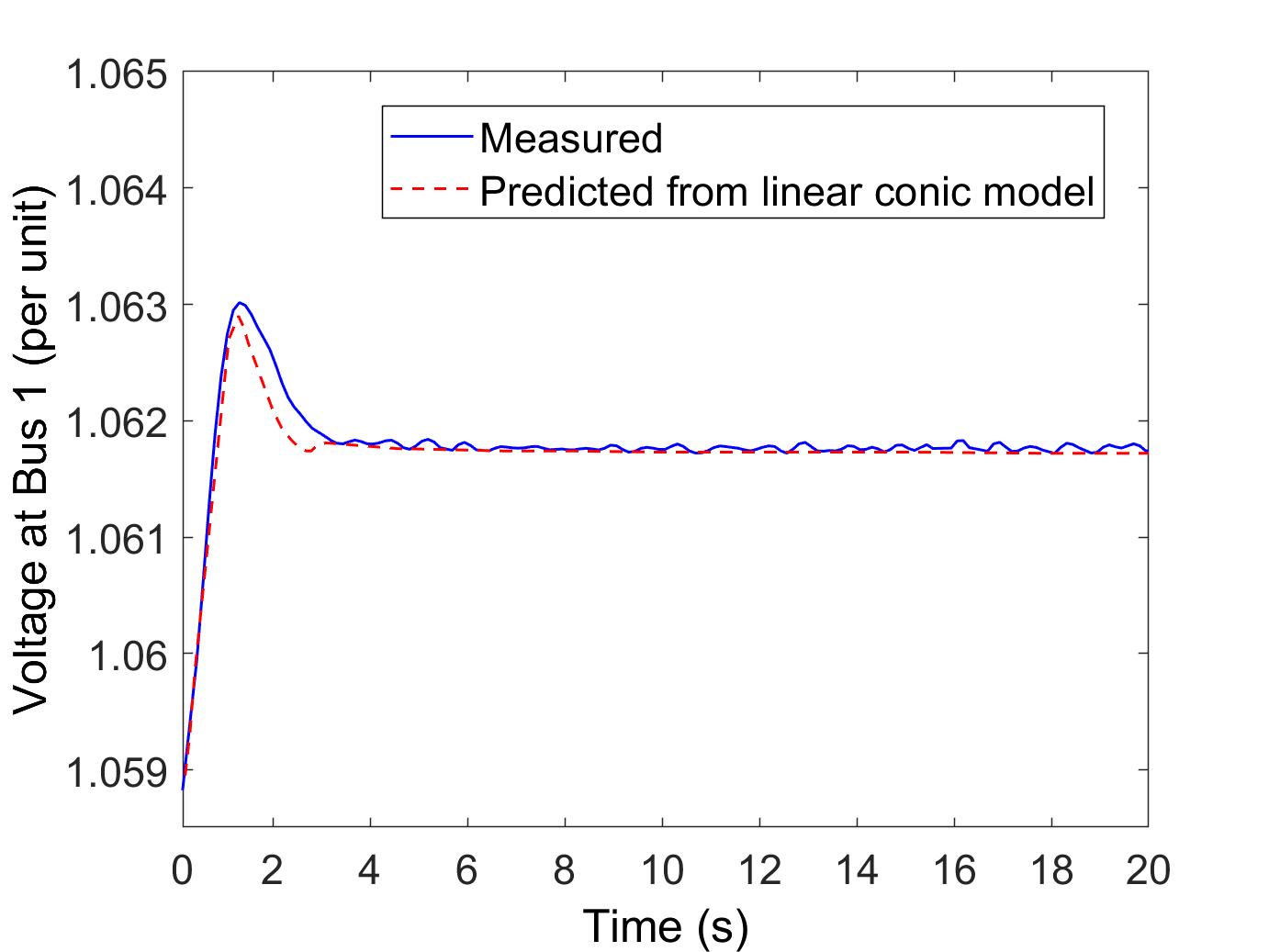}
		\caption{{Comparison of linear dissipative (conic) model and nonlinear system for 14-bus test microgrid.\vspace{-1.2em}}}
		\label{fig:test_compare}
	\end{figure}
	\section{Conclusion}%\label{sec:case}
	We considered the problem of identifying a dissipative linear model of an unknown nonlinear system from time-domain input-output data, when a baseline linear model of the system can be easily obtained using the physics of the system and/or standard system identification techniques. We propose a technique to perturb the system matrices of the baseline model to obtain a strictly dissipative linear model that closely approximates the original nonlinear system. While the proposed approach is offline, it is promising to extend the perturbation approach to quickly identify dissipative models in an online setting, where a baseline model is typically already available. 
	
	\bibliographystyle{IEEEtran}
	\bibliography{references} 
	
	\balance
\end{document}